\documentclass[a4paper,UKenglish]{lipics-v2018}

\usepackage{todonotes}
\usepackage[utf8]{inputenc}
\usepackage{multicol}
\usepackage{algorithmic}
\usepackage[ruled,vlined]{algorithm2e}
\usepackage{txfonts}
\usepackage{xspace}
\usepackage{hyperref}
\usepackage{soul}
\usepackage{xspace}

\usepackage{graphicx}

\def\RR{{\mathbb R}}

\usepackage{amsmath}
\usepackage[utf8]{inputenc}
\usepackage{soul}

\newcommand{\OO}{\mathcal{O}\xspace}
\newcommand{\sO}{\widetilde{\mathcal{O}}\xspace}

\newcommand{\volesti}{\href{https://github.com/GeomScale/volume_approximation}{\textcolor{blue}{\texttt{volesti}}}\xspace}
\newcommand{\eigen}{\href{http://eigen.tuxfamily.org}{\textcolor{blue}{\texttt{eigen}}}\xspace}
\newcommand{\boost}{\href{http://boost.com}{\textcolor{blue}{\texttt{boost}}}\xspace}
\newcommand{\mosek}{\href{https://www.mosek.com/}{\textcolor{blue}{\texttt{mosek}}}\xspace}

\newtheorem{thm}{Theorem}
\newtheorem{lem}[thm]{Lemma}

\title{Geometric algorithms for sampling
  the flux space of
  metabolic networks} 

\titlerunning{Geometric analysis of metabolic networks} 

\author{Apostolos Chalkis}{Department of Informatics \& Telecommunications \\ National \& Kapodistrian University of Athens, and\\
Athena Research Innovation Center, Greece}{achalkis@di.uoa.gr}{}{}

\author{Vissarion Fisikopoulos}{Department of Informatics \& Telecommunications \\ National \& Kapodistrian University of Athens, Greece}{vfisikop@di.uoa.gr}{}{}

\author{Elias Tsigaridas}{Inria Paris and  IMJ-PRG,\\ Sorbonne Universit\'e and Paris Universit\'e}{elias.tsigaridas@inria.fr}{}{}

\author{Haris Zafeiropoulos}{Department of Biology, University of Crete \\ Institute of Marine Biology, Biotechnology and Aquaculture, Hellenic Centre for Marine Research}{haris-zaf@hcmr.gr}{}{}

\authorrunning{Chalkis et al.} 

\Copyright{Apostolos Chalkis and Vissarion Fisikopoulos and Elias Tsigaridas and Haris Zafeiropoulos } 

\subjclass{Mathematics of computing$\rightarrow$Mathematical software; Applied computing$\rightarrow$\\Systems biology; Computing methodologies$\rightarrow$Modeling and simulation}


\category{} 

\relatedversion{} 

\supplement{}

\acknowledgements{We would like to thank the anonymous reviewers for their helpful comments and suggestions. We also thank Ioannis Emiris for his useful comments.}

\nolinenumbers
\hideLIPIcs

\begin{document}

\maketitle

\begin{abstract}
  Systems Biology is a fundamental field and paradigm that introduces a new era in Biology.
  The crux of its functionality and usefulness relies on metabolic networks
  that model the reactions occurring inside an organism
  and provide the means to understand the underlying mechanisms that govern biological systems.
  Even more, metabolic networks have a broader impact that ranges from
  resolution of ecosystems to personalized medicine.

  The analysis of metabolic networks is a computational geometry oriented field
  as one of the main operations they depend on is sampling uniformly points from  polytopes;
  the latter provides a representation of the steady states of the metabolic networks.
  However, the polytopes that result from biological data are of very high dimension (to the order of thousands) and in most, if not all, the cases are considerably skinny.
  Therefore, to perform uniform random sampling efficiently in this setting, we need
  a novel algorithmic and computational framework specially tailored
  for the properties of metabolic networks.

  We present a complete software framework to handle sampling in metabolic networks.
  Its backbone is a Multiphase Monte Carlo Sampling (MMCS) algorithm
  that unifies rounding and sampling in one pass, obtaining both upon termination.
  It exploits an
  improved variant of the Billiard Walk that enjoys faster arithmetic complexity per step.
  We demonstrate the efficiency of our approach by performing extensive experiments
  on various metabolic networks.
  Notably, sampling on the most complicated human metabolic network accessible today, Recon3D,
  corresponding to a polytope of dimension  $5\,335$, took less than $30$ hours.
  To our knowledge, that is out of reach for existing software.
\end{abstract}


\section{Introduction}
\label{sec:intro} 

\subsection{The field of Systems Biology}

Systems Biology establishes a scientific approach and a paradigm. As a
research approach, it is the qualitative and quantitative study of the systemic
properties of a biological entity along with their ever evolving interactions
\cite{ klipp2016systems, kohl2010systems}.
By combining experimental studies  with mathematical
modeling it analyzes the function and the behavior of biological systems.
In this setting, we model the interactions between the  components of a system
to shed light  on the system's \textit{raison d'être} and to decipher its underlying mechanisms
in terms of evolution, development, and physiology \cite{ideker2001new}.

Initially, Systems Biology emerged as a need. New technologies in Biology
accumulate vast amounts of information/data from different levels of the
biological organization, i.e., genome, transcriptome, proteome, metabolome
\cite{quinn2016sample}. This leads to the emerging question \textit{"what shall
  we do with all these pieces of information"?} The answer, if we consider
Systems Biology as a paradigm, is to move away from reductionism, still the main
conceptual approach in biological research, and adopt holistic approaches for
interpreting how a system's properties emerge~\cite{noble2008music}. The
following diagram provides a first, rough, mathematical formalization of this
approach.

\begin{center}
\textit{components} $\,\to\,$ \textit{networks} $\,\to\,$ \textit{in silico models} $\,\to\,$\textit{phenotype} \cite{palsson2015systems}. \\ \end{center}

Systems Biology expands in all the different levels of living entities, from the
molecular, to the organismal and ecological level. The notion that
penetrates all  levels horizontally is \emph{metabolism}; the
process that modifies molecules and  maintains the living state of a
cell or an organism through a set of chemical reactions
\cite{schramski2015metabolic}. The reactions begin with a particular molecule
which they convert into some other molecule(s), while they are catalyzed by
enzymes in a key-lock relationship.
We call the quantitative relationships between the components of a reaction   \emph{stoichiometry}.
Linked reactions, where the product of the first acts as the substrate for the
next, build up metabolic pathways. Each pathway is responsible for a certain
function. We can link together the aggregation of all the pathways that take
place in an organism (and their corresponding reactions)
and represent them mathematically using  the reactions' stoichiometry.
Therefore, at the species level, metabolism is a network of its metabolic pathways and we call
these representations \emph{metabolic networks}.

\subsection{From metabolism to computational geometry}

The complete reconstruction of the metabolic network of an organism is a
challenging, time consuming, and computationally intensive task; especially for species of high level of complexity such as \emph{Homo sapiens}.
Even though sequencing the complete genome of a species is becoming a trivial task
providing us with quality insight, manual curation is still mandatory and large groups 
of researchers need to spend a great amount of time to build such models \cite{thiele2010protocol}.
However, over the last few years, automatic reconstruction approaches for building genome-scale metabolic 
models \cite{machado2018fast} of relatively high quality have been developed.
Either way, we can now obtain the metabolic network of a bacterial species (single cell species)
of a tissue and even the complete metabolic network of a mammal.
Biologists are also moving towards obtaining such networks for all the species present in a microbial community. This will allow us to further investigate the dynamics, the functional profile, and the inter-species reactions that occur.
Using the stoichiometry of each reaction, which is always the same in the various species,
we convert the metabolic network of an organism to a mathematical model.
Thus, the metabolic network becomes an \emph{in silico} model of the knowledge it represents.
In metabolic networks analysis mass and energy are considered to be conserved
\cite{palsson2009metabolic}. As many homeostatic states, that is steady internal
conditions \cite{shishvan2018homeostatic}, are close to steady states (where the
production rate of each metabolite equals its consumption rate
\cite{cakmak2012new}) we commonly use the latter in metabolic networks analysis.

Stoichiometric coefficients are the number of molecules a biochemical reaction
consumes and produces. The coefficients of all the reactions in a network,
with $m$ metabolites and $n$ reactions ($m < n$),  form
the \emph{stoichiometric matrix} $S\in \RR^{m\times n}$~\cite{palsson2015systems}.
%
The nullspace of $S$ corresponds to the steady states of the network:
\begin{equation}
  \label{eq:Sv}
S \cdot x = 0 ,
\end{equation}
where $x \in \RR^n $ is the \textit{flux vector} that contains  the fluxes
of each chemical reaction of the network.
Flux is the rate of turnover of molecules through a metabolic pathway.
\begin {figure}[t!]
    \centering
    \includegraphics[scale=0.166]{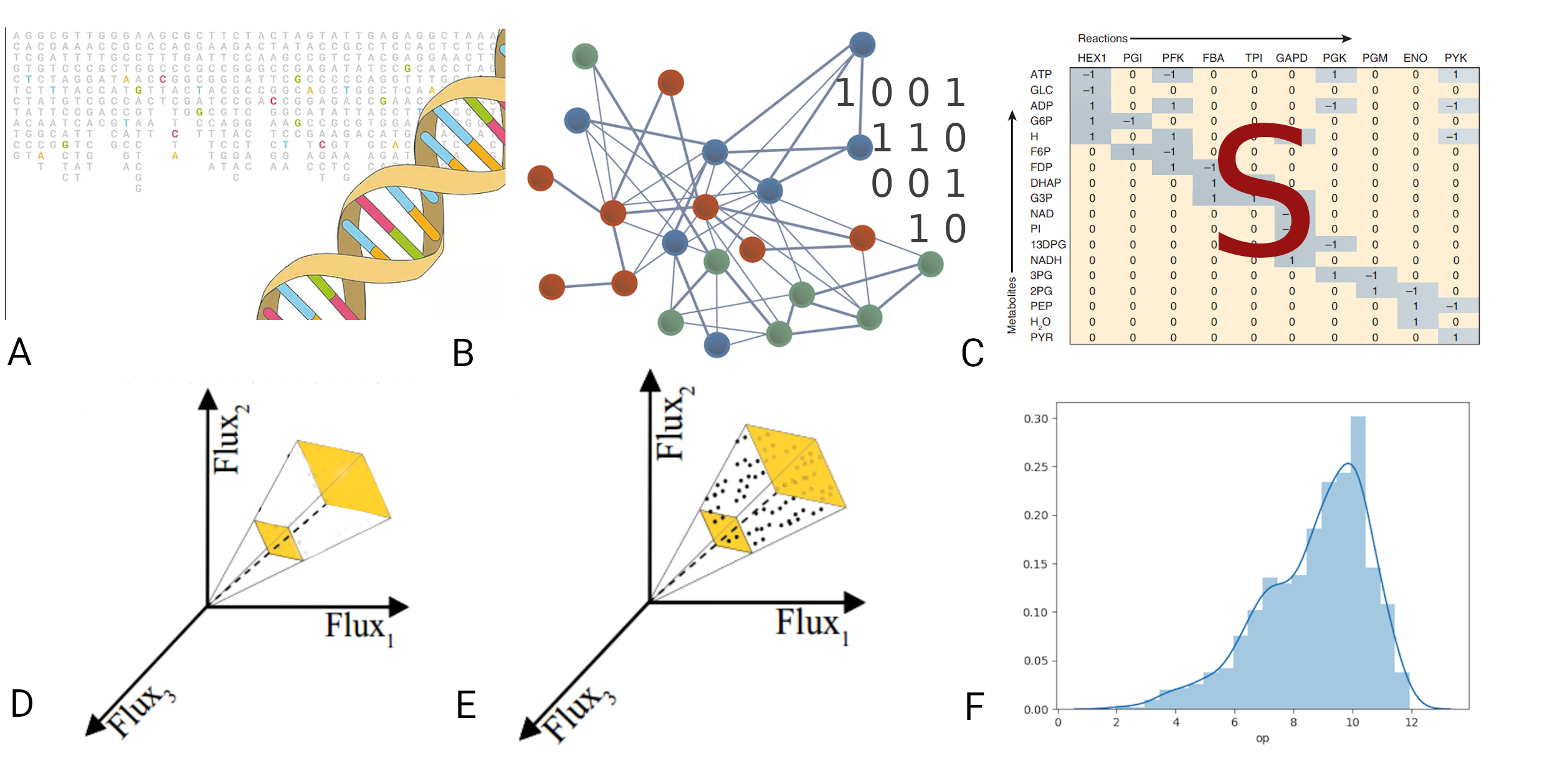}
    \caption{
    From DNA sequences to distributions of metabolic fluxes.
    (A) The genes of an organism provide us with the enzymes that it can potentially produce. Enzymes are like a blueprint for the reactions they can catalyze.
    (B) Using the enzymes we identify the reactions in the organism.
    (C) We construct the stoichiometric matrix of the metabolic  model.
    (D) We consider the flux space under different conditions (e.g., steady states);
    they correspond  to  polytopes containing flux vectors addressing these conditions.
    (E) We sample from polytopes that are typically skinny and of high dimension.
    (F) The distribution of the flux of a reaction provides  great insights
    to biologists.
    }
    \label{fig:workflow}
\end{figure}

All physical variables are finite, therefore  the flux (and the concentration)
is bounded  \cite{palsson2015systems}; that is
for each coordinate $x_i$ of the $x$, there are $2n$ constants
$x_{ub, i}$ and  $x_{lb, i}$
such that  $x_{lb,i} \le x_i \le x_{ub, i}$, for $i  \in [n]$.
We derive the constraints from explicit experimental information.
In cases where there is no such information, reactions are left unconstrained by
setting arbitrary large values to their corresponding bounds according to their reversibility properties; i.e., if a reaction is reversible then its flux might be negative as well~\cite{lularevic2019improving}.
The constraints define a $n$-dimensional box
containing both the steady and the dynamic states of the system.
If we intersect that box with the nullspace of
$S$, then we define a polytope that encodes all the possible steady
states and their  flux distributions \cite{palsson2015systems}.
We call it the  steady-state \emph{flux space}.
Fig.~\ref{fig:workflow} illustrates the complete workflow from building a metabolic network to the computation of a flux distribution.

Using the polytopal representation, a commonly used method for the analysis of a metabolic network is Flux Balance Analysis (FBA)~\cite{orth2010flux}.
FBA identifies a single optimal flux distribution by optimizing a linear objective function over a polytope~\cite{orth2010flux}. Unfortunately, this is a \textit{biased} method because it depends on the  selection of the objective function.
To study the  global features of a metabolic network we need \emph{unbiased methods}. To obtain  an accurate picture of the whole solution space we exploit sampling techniques \cite{schellenberger2009use}.
If collect a sufficient number of points uniformly distributed in the interior of the polytope, then the biologists can study the properties of certain components of the whole network and deduce significant biological insights~\cite{palsson2015systems}. Therefore, efficient sampling tools are of great importance.

\begin{figure}[t!]
    \centering
    \includegraphics[scale=0.1645]{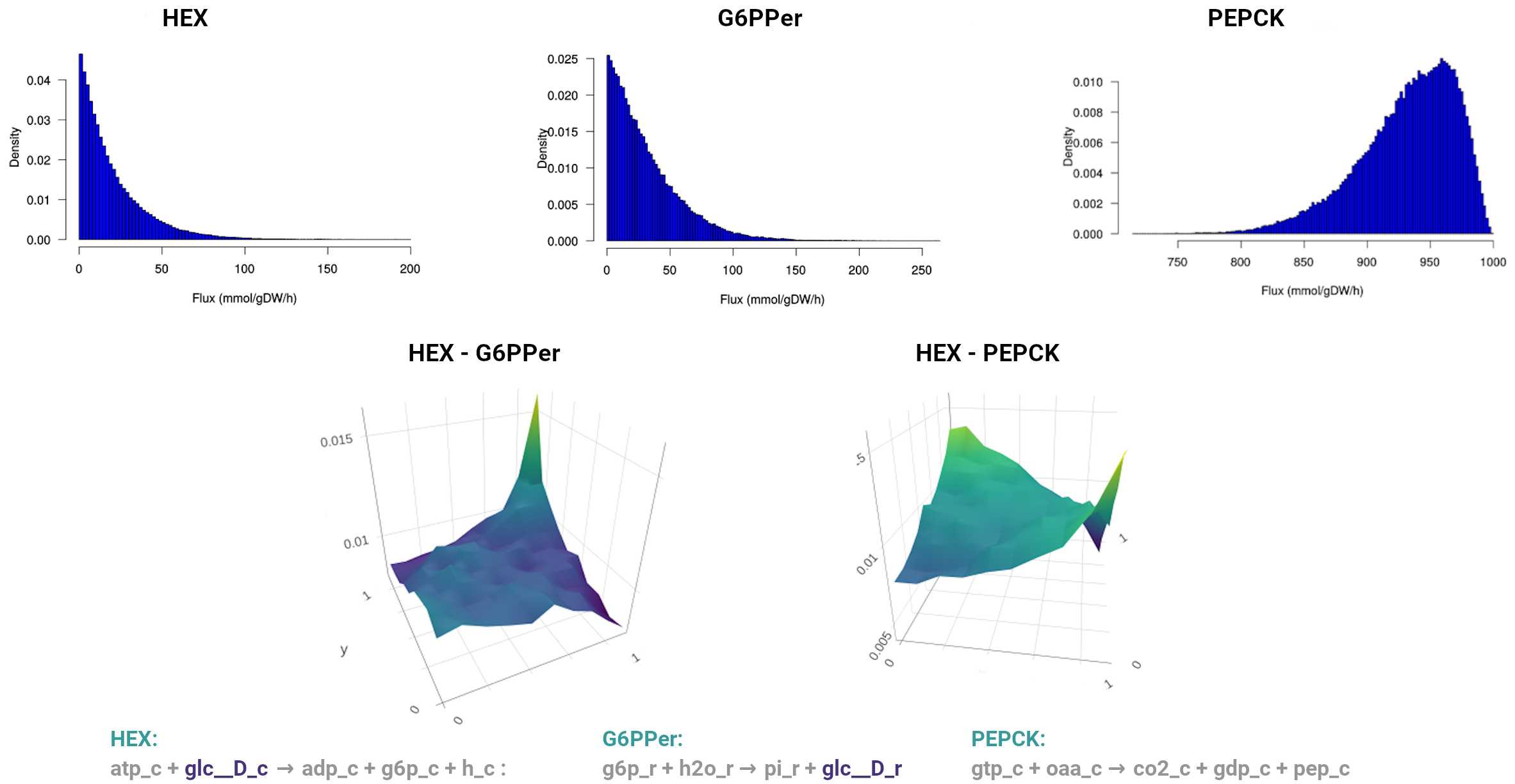}
    \caption{Flux distributions in the most recent human metabolic network Recon3D~\cite{brunk2018recon3d}. We estimate the flux distributions of the reactions catalyzed by the enzymes
     Hexokinase (D-Glucose:ATP) (HEX), Glucose-6-Phosphate Phosphatase,
    Edoplasmic Reticular (G6PPer)
    and Phosphoenolpyruvate carboxykinase (GTP) (PEPCK).
    As we sample steady states, the production rate of $glc\_\_D \_c$ should be equal to its consumption rate. 
    Thus, in the corresponding copula, we see a positive dependency between HEX,
    i.e., the reaction that consumes $glc\_\_D \_c$ and G6PPer, that produces it.
    Furthermore, the PEPCK reaction operates when there is no $glc\_\_D\_c$ available
    and does not operate when the latter is present.
    Thus, in their copula we observe a negative dependency between HEX and PEPCK.
    A copula is a bivariate probability distribution for which the marginal probability distribution of each variable is uniform.
    It implies a positive dependency when the mass of the distribution concentrates along the up-diagonal (HEX - G6PPer)
    and a negative dependency when the mass is concentrated along the down-diagonal (HEX - PEPCK). 
    The bottom line contains the  reactions and their stoichiometry.
    }
    \label{fig:copulas}
\end{figure}

\subsection{Metabolic networks through the lens of random sampling} \label{sec:previous_work}

Efficient uniform random sampling on polytopes resulting from metabolic
networks is a very challenging task  both from the theoretical (algorithmic)
and the engineering (implementation)  point of view.
First, the dimension of the polytopes is of the order of certain thousands. This
requires, for example, advanced engineering techniques to cope with memory
requirements and to perform linear algebra operations with large matrices;
e.g., in
Recon3D~\cite{brunk2018recon3d} we compute the null space of a
$8\,399 \times 13\,543$ matrix. Second, the polytopes are rather skinny
(Sec.~\ref{sec:experiments}); this makes it harder for sampling algorithms to
move in the interior of polytopes and calls for novel practical techniques to
sample.

There is extended on-going research concerning advanced algorithms and
implementations for sampling metabolic networks over the last decades. Markov
Chain Monte Carlo algorithms such as Hit-and-Run (HR)~\cite{smith84} have been
widely used to address the challenges of sampling. Two variants of HR are the
non-Markovian Artificial Centering Hit-and-Run
(ACHR)~\cite{kaufman1998direction} that has been widely used in sampling
metabolic models, e.g.,~\cite{Saa16}, and Coordinate Hit-and-Run with Rounding
(CHRR)~\cite{haraldsdottir2017chrr}. 
The latter is part of the {\tt cobra} toolbox~\cite{heirendt2019creation}, the most commonly used software package for the analysis of metabolic networks. 
CHRR enables sampling from complex metabolic
network corresponding to the highest dimensional polytopes so far. There are
also stochastic formulations where the inclusion of experimental noise in the
model makes it more compatible with the stochastic nature of biological
networks~\cite{MacGillivray17}. The recent study in~\cite{fallahi2020comparison}
offers an overview as well as an experimental comparison of the currently
available samplers.

These implementations played a crucial role in actually performing in practice
uniform sampling from the flux space. However, they are currently limited to
handle polytopes of dimension say $\leq 2\, 500$ \cite{fallahi2020comparison,
  haraldsdottir2017chrr}. This is also the order of magnitude of the most
complicated, so far, metabolic network model built,
Recon3D~\cite{brunk2018recon3d}. By including $13\,543$ metabolic reactions and
involving $4 \,140$ unique metabolites, Recon3D provides a representation of the
$17\%$ of the functionally of annotated human genes. To our knowledge, there is
no method that can efficiently handle sampling from the flux space of Recon3D.

Apparently, the dimension of the polytopes will keep rising and not only for the ones corresponding to  human metabolic networks.
Metabolism governs systems biology at all its levels, including the one of the community.
Thus, we are not only interested in sampling a sole metabolic network, even if it has the challenges of the human.
Sampling in polytopes associated to network of networks are the next big thing in metabolic networks analysis and in Systems Biology \cite{bernstein2019metabolic,perez2016metabolic}.

Regarding the sampling process, from the theoretical point of view, we are
interested in the convergence time, or {\it mixing time}, of the Markov Chain,
or geometric {\it random walk}, to the target distribution. Given a
$d$-dimensional polytope $P$, the mixing time of several geometric random walks
(e.g., HR or Ball Walk) grows quadratically with respect to the sandwiching
ratio $R/r$ of the polytope~\cite{LovSim,Lovasz06}. Here $r$ and $R$ are the
radii of the smallest and largest ball with center the origin that contains, and
is contained, in $P$, respectively; i.e., $rB_d \subseteq P \subseteq RB_d$,
where $B_d$ is the unit ball. It is crucial to reduce $R/r$, i.e., to put $P$ in
well a rounded position where $R/r = \sO(\sqrt{d})$; the $\sO(\cdot)$ notation
means that we are ignoring polylogarithmic factors. A powerful approach to
obtain well roundness is to put $P$ in \emph{near isotropic position}. In
general, $K \subset \RR^d$ is in isotropic position if the uniform distribution
over $K$ is in isotropic position, that is $\mathbb{E}_{X\sim K}[X] = 0$ and
$\mathbb{E}_{X\sim K}[X^TX] = I_d$, where $I_d$ is the $d\times d$ identity
matrix. Thus, to put a polytope $P$ into isotropic position one has to generate
a set of uniform points in its interior and apply to $P$ the transformation that
maps the point-set to isotropic position; then iterate this procedure until $P$
is in $c$-isotropic position \cite{ Cousins15,Lovasz06}, for a constant~$c$.
In~\cite{Adamczak10} they prove that $\OO(d)$ points suffice to achieve
$2$-isotropic position. Alternatively in \cite{haraldsdottir2017chrr} they
compute the maximum volume ellipsoid in $P$, they map it to the unit ball, and
then apply to $P$ the same transformation. They experimentally show that a few
iterations suffice to put $P$ in John's position~\cite{John48}. Moreover, there
are a few algorithmic contributions that combine sampling with distribution
isotropization steps, e.g., the multi-point walk~\cite{Bertsimas04} and the
annealing schedule~\cite{Kalai06}.

An important parameter of a random walk is the {\it walk length}, i.e., the number of the intermediate points that a random walk visits before producing a single sample point.
The longer the walk length of a random walk is, the smaller the distance of the current distribution to the stationary (target) distribution becomes.
For the majority of  random walks there are bounds on the walk length to bound the mixing time with respect to a statistical distance. For example, HR generates a sample from a distribution with total variation distance less than $\epsilon$ from the target distribution after $\sO(d^3)$~\cite{Lovasz06} steps, in a well rounded convex body and for log-concave distributions. Similarly, CDHR mixes after a polynomial, in the diameter and the dimension, number of steps~\cite{Laddha20,Narayanan20} for the case of uniform distribution. However, extended practical results have shown that both CDHR and HR converges after $\OO(d^2)$ steps~\cite{Chalkis19, Cousins15, haraldsdottir2017chrr}. The leading algorithms for uniform polytope sampling are the Riemannian Hamiltonian Monte Carlo sampler~\cite{VempalaRiem18} and the Vaidya walk~\cite{Chen18}, with mixing
times $\sO(md^{2/3})$ and $\sO(m^{1/2}d^{3/2})$ steps, respectively.
However, it is not clear if these random walks can outperform CDHR in practice,
because of their  high cost per step and numerical instability.

Billiard Walk (BW)~\cite{Gryazina14} is a random walk that employs linear
trajectories in a convex body with boundary reflections; alas with an unknown
mixing time. The closest guarantees for its mixing time are those of HR and
stochastic billiards~\cite{Dieker15}. Interestingly,~\cite{Gryazina14} shows
that, experimentally, BW converges faster than HR for a proper tuning of its
parameters. The same conclusion follows from the computation of the volume of
zonotopes~\cite{Chalkis20}. It is not known how the sandwiching ratio of $P$
affects the mixing time of BW. Since BW employs reflections on the boundary,
we can consider it as a special case of Reflective Hamiltonian Monte Carlo~\cite{ChePioCaz18}.

For almost all random walks the theoretical bounds on their mixing times are
pessimistic and unrealistic for  computations. Hence,
if we terminate the random walk earlier, we generate samples that are  usually highly correlated.
There are several {\it MCMC Convergence Diagnostics} \cite{Roy20} to check
if  the quality of a  sample can  provide an accurate
approximation of the target distribution. For a dependent sample, a powerful
diagnostic is the {\it Effective Sample Size} (ESS). It is the number of
effectively independent draws from the target distribution that the Markov chain
is equivalent to. For autocorrelated samples, ESS bounds
the uncertainty in estimates \cite{geyer92} and provides  information about
the quality of the sample. There are several statistical tests to evaluate the quality of a generated sample, e.g., potential scale reduction factor (PSRF)~\cite{Gelman92}, maximum mean discrepancy (MMD)~\cite{Gretton12},
and the uniform tests~\cite{CousinsThesis17}.
Interestingly, 
the copula representation we employ in Fig.~\ref{fig:copulas} to capture the dependence between two fluxes of reactions was also used  successfully in a geometric framework to detect financial crises capturing the dependence between portfolio return and volatility~\cite{Cales18}.

\subsection{Our contribution}
\label{sec:contributions}

We introduce a Multi-phase Monte Carlo Sampling (MMCS) algorithm
(Sec.~\ref{sec:MMCS} and Alg.~\ref{alg:MMCS}) to sample from a polytope~$P$. In
particular, we split the sampling procedure in phases where, starting from $P$,
each phase uses the sample to round the polytope. This improves the efficiency
of the random walk in the next phase, see Fig.~\ref{fig:mmcs_phases} and
Table~\ref{tab:sampling_phases}.
For sampling, we propose an improved variant of Billiard Walk (BW)
(Sec.~\ref{subsec:billiard} and Alg.~\ref{alg:billiard}) that enjoys faster
arithmetic complexity per step. We also handle efficiently the potential arithmetic inaccuracies near to the boundary, see~\cite{ChePioCaz18}.
We accompany the MMCS algorithm with a powerful MCMC diagnostic, namely the
estimation of Effective Sample Size (ESS), to identify a satisfactory
convergence to the uniform distribution.
However, our method is flexible and  we can
use any  random walk and
combination of MCMC diagnostics to decide~
convergence.

The open-source implementation of our algorithms\footnote{\url{https://github.com/GeomScale/volume_approximation/tree/socg21}} provides a
complete software framework to handle efficiently sampling in metabolic
networks. We demonstrate the efficiency of our tools by performing experiments
on almost all the metabolic networks that are publicly available and by
comparing with the
state-of-the-art software packages as {\tt cobra}
(Sec.~\ref{subsec:experiments}). Our implementation is faster than {\tt cobra}
for low dimensional models,  with a speed-up that ranges from $10$ to $100$ times;
this gap on running times increases for bigger models
(Table~\ref{tab:results1}). The quality of the sample our software produces is
measured with two widely used diagnostics, i.e., ESS and potential scale reduction factor (PSRF)~\cite{Gelman92}. The highlight of
our method is the ability to sample from the most complicated human metabolic
network that is accessible today, namely Recon3D. In Fig.~\ref{fig:copulas} we estimate marginal univariate and bivariate flux distributions in Recon3D which validate
(a) the  quality of the sample by confirming a mutually exclusive pair of biochemical pathways,
and that (b) our method indeed generates steady states. In particular, our software
can sample $1.44\cdot 10^5$ points from a $5\,335$-dimensional polytope in a day
using modest hardware. This set of points suffices for the majority of systems
biology analytics.
To our understanding this task is out of reach for existing software.

Last, MMCS algorithm is quite general sampling scheme and
so it has the potential to address other
hard computational problems like  multivariate integration and volume estimation of polytopes.

\section{Efficient Billiard walk}\label{subsec:billiard}

The geometric random walk of our choice
to sample from a polytope
  is based on  Billiard Walk (BW)~\cite{Gryazina14},
which we modify to reduce the per-step cost.

For a polytope $P = \{ x \in \RR^d \,|\, A x \leq b  \}$,
where $A \in \RR^{k \times d}$ and $b \in \RR^k$,
BW starts from a given point $p_0 \in P$,
selects uniformly at random a
direction, say $v_0$, and it moves along the direction of $v_0$ for length $L$;
it reflects on the boundary if necessary. This results a new point $p_1$ inside
$P$. We repeat the procedure from $p_1$. Asymptotically it converges to the
uniform distribution over $P$. The length is $L=-\tau \ln \eta$, where $\eta$ is
a uniform number in $(0,1)$, that is $\eta\sim\mathcal{U}(0,1)$, and $\tau$ is a
predefined constant. It is useful to set a bound, say $\rho$, on the number of
reflections to avoid computationally hard cases where the trajectory may stuck
in corners. In \cite{Gryazina14} they set
$\tau \approx \mathop{diam}(P)$ and $\rho =10d$.
Our choices for $\tau$ and $\rho$ depend on a
burn-in step that we detail in  Sec.~\ref{sec:experiments}.

At each step of BW we
compute the intersection point of a ray, say $\ell:=\{p+t\varv,\ t\in\RR_+ \}$,
with the boundary of $P$, $\partial P$, and the normal vector of the tangent
plane at the intersection point.
The inner vector of the facet that the intersection  point belongs to is a row of $A$.
To compute the point $\partial P\cap\ell$ where the first reflection of a BW
step takes place, we solve the following $m$ linear equations
\begin{equation}\label{eq:boundary_oracle}
  a_j^T(p_0 + t_j\varv_0) = b_j \Rightarrow t_j = (b_j - a_j^Tp_0) / a_j^T\varv_0,\ j \in[k],
\end{equation}
and keep the smallest positive $t_j$; $a_j$ is the $j$-th row of the matrix $A$.
We solve each equation in $\OO(d)$ operations
and so the overall complexity is $\OO(dk)$.
A straightforward approach for BW would consider that each reflection costs $\OO(kd)$ and thus the per step cost is $\OO(\rho kd)$.
However, our improved version performs more efficiently both {\it point} and {\it direction updates} by storing computations from the previous iteration combined with a preprocessing step. The preprocessing step involves the normal vectors of the facets, that takes $m^2 d$ operations,
and the amortized per-step complexity of BW becomes $\OO((\rho + d)k)$.
The pseudo-code appears in Alg.~\ref{alg:billiard} in the appendix.

\begin{lem}
  \label{lem:BW-step-cost}
  The amortized per step complexity of BW (Alg.~\ref{alg:billiard}) is $\OO((\rho + d)k)$ after a preprocessing step that takes $\OO(k^2d)$ operations, where $\rho$ is the maximum number of reflections per step.
\end{lem}
\begin{proof}
  The first reflection of a BW  step costs $\OO(kd)$.
  During its computation, we store all the values of the inner
  products $a_j^Tx_0$ and $a_j^T\varv_0$.
  At the reflection $i>0$, we start from a point $x_i$,
  and the solutions of the corresponding linear equations are
  \begin{equation}
    \label{eq:billiard_implementation1}
    \begin{split}
      a_j^T(p_i + t_j\varv_i) = b_j & \Rightarrow
      a_j^T(p_{i-1} +  t_{i-1}\varv_{i-1})
      + t_ja_j^T(\varv_{i-1}-2(\varv_{i-1}^Ta_r) a_r) = b_j \\
      & \Rightarrow  t_j = \frac{b_j - a_j^T(p_{i-1} +  t_{i-1}\varv_{i-1})}{a_j^T(\varv_{i-1}-2(\varv_{i-1}^Ta_r) a_r)},\
      \text{ for }  j \in [k]  ,
    \end{split}
  \end{equation}
  \begin{equation}\label{eq:billiard_implementation2}
    \text{ and } \, \varv_{i+1} = \varv_i -2(\varv_i^Ta_l) a_l,
  \end{equation}
  where $a_r,\ a_l$ are the normal vectors of the facets that $\ell$ hits at reflection $i-1$ and $i$ respectively, and $t_{i-1}$ the solution of the  reflection $i-1$.
  The index $l$ of the normal $a_l$
  corresponds to the equation with the smallest positive $t_j$
   in~(\ref{eq:billiard_implementation1}).
   We solve each of the  equations in (\ref{eq:billiard_implementation2}) in $\OO(1)$ based on our bookkeeping from the  previous reflection.
   We also store the inner product $\varv_i^Ta_l$ in (\ref{eq:billiard_implementation2}) from the  previous reflection.
   After computing all $a_i^Ta_j$ as a preprocessing step, which takes $k^2d$ operations, the total per-step cost of Billiard Walk is $\OO((d+\rho)k)$.
\end{proof}

The use of floating point arithmetic could result to points outside $P$ due to rounding errors when computing boundary points. To avoid this, when we compute the roots in Equation~(\ref{eq:boundary_oracle}) we exclude the facet that the ray hit in the previous reflection.

\section{Multiphase Monte Carlo Sampling algorithm}
\label{sec:MMCS}

To sample steady states in the flux space of a metabolic network, with $m$
metabolites and $n$ reactions, we introduce a Multiphase Monte Carlo Sampling
(MMCS) algorithm; it is multiphase because it consists of a sequence of sampling
phases.

Let $S\in\mathbb{R}^{m\times n}$ be the  stoichiometric matrix
and  $x_{lb},\ x_{ub}\in\mathbb{R}^{n}$
bounds on the fluxes.
The flux space is the bounded convex polytope
\begin{equation}
  \label{eq:steady_states}
  \mathrm{FS} :=  \{x \in \RR^n \,|\, Sx = 0,\ x_{lb}\leq x \leq x_{ub} \} \subset \RR^n .
\end{equation}
The dimension, $d$, of $\mathrm{FS}$ is smaller than the dimension of the
ambient space; that is $d \leq n$.
To work with a full dimensional polytope we restrict the box
induced by the inequalities $x_{lb}\leq x \leq x_{ub}$ to the null space of  $S$.
Let the  H-representation of the box be
$\left\{ x \in \RR^n \,\Big |\, \begin{pmatrix} I_n \\ -I_n\end{pmatrix} x
  \leq \begin{pmatrix} x_{ub} \\ x_{lb} \end{pmatrix} \right\}$,
where $I_n$ is the $n\times n$ identity matrix,
and let $N\in\RR^{n\times d}$ be the matrix of the null space of $S$,
that is $S \, N = 0_{m \times d}$.
Then
$P = \{ x\in\RR^d\ |\ A x \leq b \}$,
where $A = \begin{pmatrix} I_n N \\ -I_n N\end{pmatrix}$
and $b = \begin{pmatrix} x_{ub} \\ x_{lb} \end{pmatrix}  N$,
is a
full dimensional polytope (in $\RR^d$).
After we sample (uniformly) points from $P$,
we transform them to uniformly distributed points (that is steady states) in $\mathrm{FS}$
by applying  the linear map induced by $N$.

\vspace{8pt}
MMCS generates, in a sequence of sampling phases, a set of points,
 that is almost equivalent to $n$ independent uniformly distributed points in $P$, where $n$ is given.
At each phase, it employs Billiard Walk (Section~\ref{subsec:billiard}) to sample approximate uniformly distributed points, rounding
to speedup sampling,
and uses the Effective Sample Size (ESS) diagnostic to decide termination.
The pseudo-code of the algorithm appears in Alg.~\ref{alg:MMCS}.

\emph{Overview.} Initially we set $P_0 = P$.

At each phase $i \geq 0$ we sample at most $\lambda$ points from $P_i$. We
generate them in chunks; we also call them  \emph{chain} of sampling points.
Each chain contains at  most $l$ points (for simplicity consider $l = \OO(1)$). To
generate the points in each chain we employ BW, starting from a point inside
$P_i$; the starting point is different for each chain.
We repeat this procedure until the total number of samples in $P_i$ reaches the maximum number $\lambda$; we need $\tfrac{\lambda}{l}$ chains.
To compute a starting point for a chain, we pick a point uniformly at random in the
Chebychev ball of $P_i$ and we perform $\OO(\sqrt{d})$
burn-in BW steps to obtain a warm start.

After we have generated  $\lambda$ sample points we perform
a rounding step on $P_i$ to obtain the polytope of the next phase, $P_{i+1}$.
In particular, we compute a linear transformation, $T_i$, that puts the sample
into isotropic position and then $P_{i+1} = T_i(P_i)$.
The efficiency of BW improves from one phase to the next one
because the sandwiching ratio decreases and so
the  average number of reflections decreases
and thus the convergence  to the uniform distribution accelerates
(Section~\ref{subsec:experiments}).
That is we obtain faster a sample of better quality.
Finally, the (product of the) inverse transformations maps the samples to $P_0 = P$. Fig.~\ref{fig:mmcs_phases}
depicts the procedure.

\emph{Termination.} There are no bounds on the mixing time of BW \cite{Gryazina14},
hence for termination we rely on ESS.
MMCS terminates when the minimum ESS among all the univariate
marginals is larger than a requested value. We chose the marginal
distributions (of each flux) because they are essential for systems biologists,
see \cite{Bordel10} for a typical example.
In particular, after we generate a chain, the algorithm updates the ESS of each
univariate marginal to take into account all the points that we have sampled in $P_i$,
including the newly generated chain. We keep the minimum, say $n_i$, among all
marginal ESS values. If $\sum_{j=0}^in_j$ becomes larger than $n$ before the
total number of samples in $P_i$ reaches the upper bound $\lambda$, then  MMCS
terminates. Otherwise, we proceed to the next phase. In summary, MMCS terminates
when the sum of the minimum marginal ESS values of each phase reaches $n$.

\begin{figure}[t!]
    \centering
    \includegraphics[scale=0.153]{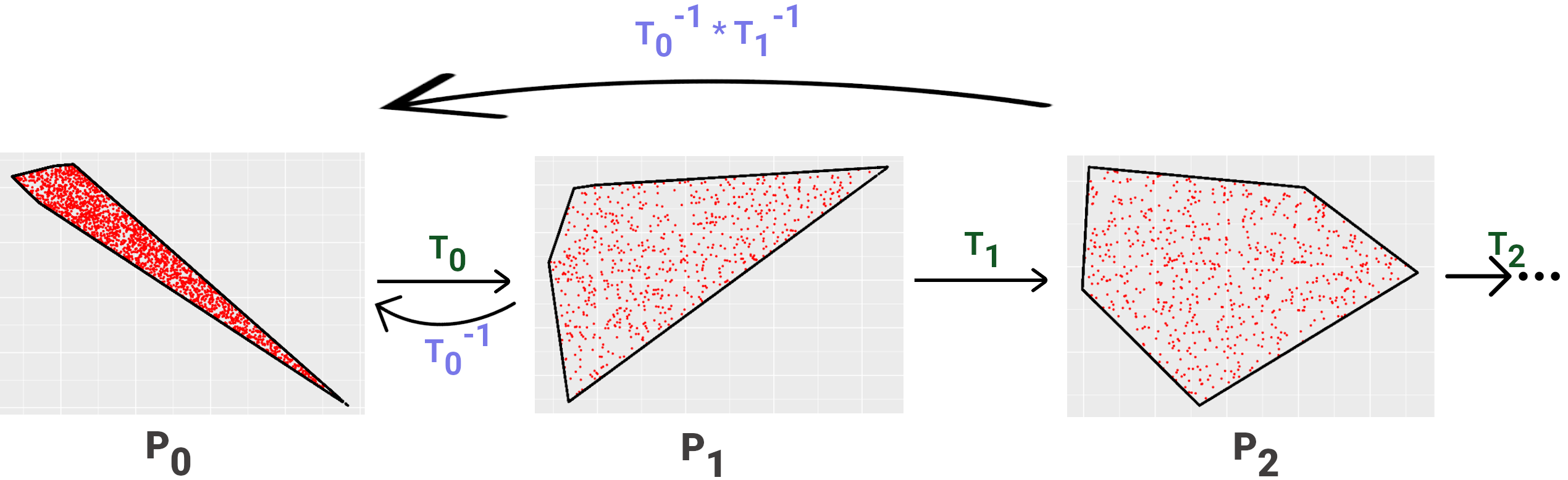}
    \caption{An illustration of our Multiphase Monte Carlo Sampling algorithm. The method is given an integer $n$ and starts at phase $i=0$ sampling from $P_0$. In each phase it samples a maximum number of points $\lambda$. If the sum of Effective Sample Size in each phase becomes larger than $n$ before the total number of samples in $P_i$ reaches $\lambda$ then the algorithm terminates. Otherwise, we proceed to a  new phase.
    We map back to $P_0$ all the generated samples of each phase.}
    \label{fig:mmcs_phases}
\end{figure}

\begin{algorithm}[t]
  \caption{Multiphase Monte Carlo Sampling$(P, n, l, \lambda, \rho, \tau, W)$}
	\label{alg:MMCS}
	\SetKwInOut{Input}{Input}
	\SetKwInOut{Output}{Output}
    \SetKwInOut{Require}{Require}
    \SetKwRepeat{Do}{do}{while}

    \Input{A full dimensional polytope $P\in\RR^d$; Requested effectiveness $n\in\mathbb{N}$; $l$ length of each chain; upper bound for the number of  generated points in each phase $\lambda$; upper bound on the number of reflections $\rho$; length of trajectory parameter $\tau$; walk length $W$.}


	\Output{A set of approximate uniformly distributed points $S\in P$}

    \BlankLine
    Set $P_0 \leftarrow P,\ sum\_ess\leftarrow 0,\ S\leftarrow\emptyset,\ i\leftarrow 0,\ T_0 = I_d$\;
    \Do{$sum\_ess < n$} {
    	$sum\_point\_phase\leftarrow 0,\ U\leftarrow\emptyset$\;
    	\Do{$sum\_point\_phase < \lambda$} {
    		Generate a starting point $p\in P_i$\;
    		Generate a set $Q$ of $l$ points with Billiard Walk starting from p\;
    		$S\leftarrow S\cup T_i^{-1}(Q)$\;
    		$U\leftarrow U\cup Q$\;
    		$sum\_point\_phase\leftarrow sum\_point\_phase + l$\;
    		Update ESS $n_i$ of this phase\;
    		\lIf{$sum\_ess + n_i \geq n$} {
    			\textbf{break}
    		}
    	}
    	$sum\_ess\leftarrow sum\_ess + n_i$\;
    	$i\leftarrow i+1$\;
    	Compute $T$ such that $T(U)$ is in isotropic position\;
    	$T_i \leftarrow T_{i-1}\circ T$\;

    }

    \KwRet $S$\;
\end{algorithm}

\emph{Rounding step}. This step is
motivated by the theoretical result in~\cite{Adamczak10} and the rounding
algorithms~\cite{Lovasz06, Cousins15}.
We apply the linear transformation $T_i$ to $P_i$
so that the sandwiching ratio of $P_{i+1}$ is smaller than that of $P_i$.
To find the suitable $T_i$ we compute  the SVD decomposition of the matrix that contains the sample row-wise \cite{Shiri20}.

\emph{Updating the  Effective Sample Size}.
The effective sample size of a sample of points generated by a process with autocorrelations $\rho_t$ at lag $t$
is function (actually an infinite series) in the $\rho_t$'s; its exact value is unknown.
Following \cite{geyer92}, we efficiently compute ESS employing a finite sum of monotone estimators
$\hat{\rho}_t$ of the autocorrelation at lag $t$, by exploiting Fast Fourier Transform. 
Furthermore, given $M$ chains of samples, the autocorrelation estimator $\hat{\rho}_t$ is given by,
\begin{equation}\label{eq:autocorrelation}
\hat{\rho}_t = 1 - \frac{C - \frac{1}{M}\sum_{i=1}^{M}\hat{\rho}_{t,i}}{B},
\end{equation}
where $C$ and $B$ are the within-sample variance estimate and the multi-chain variance estimate given in \cite{Gelman92} and $\hat{\rho}_{t,i}$ is an estimator of the autocorrelation of the $i$-th chain at lag $t$.
To update the ESS, for every new chain of points the algorithm generates, we compute
$\hat{\rho}_{t,i}$. Then, using Welford's algorithm we
update the average of the estimators of autocorrelation at lag $t$,
as well as the between-chain variance and the within-sample variance estimators given in~\cite{Gelman92}.
Finally, we update the ESS using these estimators.

\begin{lem}\label{lem:mmcs_complexity}
Let $P=\{x\in\RR^d\ |\ Ax\leq b\}$, $A\in\RR^{k\times d},\ b\in\RR^k$ a full dimensional polytope in $\RR^d$.
The total number of operations per phase that Alg.~\ref{alg:MMCS} performs, is
$\OO(W(\rho + d)k\lambda + \lambda^2d + d^3)$,
where $W$ is the walk length for Billiard Walk.
\end{lem}
\begin{proof}
The cost per step of Billiard Walk is $\OO((\rho + d)k)$.
In each phase we generate with Billiard Walk at most $\lambda$ points with walk length $W$. Thus, the cost to generate those points is $\OO(W(\rho + d)k\lambda)$.

To compute the starting point of each chain the algorithm picks a random point
uniformly distributed in the Chebychev ball of $P$ and performs $\OO(1)$
Billiard Walk steps starting from it.
The former  takes $\OO(d)$ operations and latter  takes $\OO(W(\rho + d)k)$ operations.
The total number of chains is $\OO(\lambda / l) = \OO(\lambda)$, as $l=\OO(1)$.
Thus, the total cost to generate all the starting points is $\OO(d\lambda + W(\rho + d)k\lambda)$. The update of ESS of each univariate marginal takes $\OO(1)$ operations since $l = \OO(1)$.

If the termination criterion has not been met after generating $\lambda$ points,
the algorithm computes a linear transformation to put the set of points to isotropic
position. We can do this by computing the SVD decomposition of the matrix that
contains the set of points row-wise. This corresponds to an SVD of a
$\lambda \times d$ matrix and takes $\OO(\lambda^2d + d^3)$ operations
\cite{golub13}.
\end{proof}

In Section~\ref{sec:experiments} we discuss how to tune the parameters of MMCS
to make it more efficient in practice. We also comment on the (practical)
complexity of each phase, based on the tuning.


\section{Implementation and Experiments}\label{sec:experiments}

In the sequel we present the implementation of our approach and the tuning of
various parameters. We present experiments in an extended set of BURG
models \cite{king2016bigg}, including the most complex metabolic networks i.e.,
the human Recon2D~\cite{swainston16} and Recon3D~\cite{brunk2018recon3d}. We end up to sample from polytopes of thousands
of dimensions and show that our method can estimate precisely the flux
distributions. We analyze various aspects of our method as the runtime, the
efficiency and the quality of the output. We compare our method against the
state-of-the-art software for the analysis of metabolic networks, which is the
{\tt Matlab toolbox} of {\tt cobra}~\cite{heirendt2019creation}. Our implementation for low
dimensional networks is two orders of magnitude faster than {\tt cobra}. As the
dimension grows this gap on the run-time increases. The fast mixing of billiard
walk allow us to use all the generated samples to approximate each flux
distribution improving the flux distribution estimation.

We provide a complete open-source software framework to handle big metabolic networks. The framework loads a metabolic model in some standard format (e.g., {\tt mat,\ json} files) and performs an analysis of the model e.g., compute the marginal distributions of a given metabolite. All the results in this paper are reproducible using our publicly available code\footnote{\url{https://github.com/GeomScale/volume_approximation/tree/socg21}}.
The core of our implementation is in  {\tt C++} to optimize performance while the user interface is implemented in  {\tt R}. The package employs \eigen~\cite{eigenweb} for linear algebra, \boost~\cite{boostrandom} for random number generation, \mosek~\cite{mosek} as the linear program solver, and expands \volesti~\cite{chalkis2020volesti} an open-source package for high dimensional sampling and volume approximation.
All experiments were performed on a PC with {\tt Intel® Core™ i7-6700 3.40GHz × 8 CPU} and {\tt 32GB RAM}.

\subsection{Parameter tuning for practical performance}\label{subsec:implementation}

We give details on how we tune various parameters presented in Section~\ref{sec:MMCS} in our implementation.

\textbf{Parameters of Billiard Walk}: To employ Billiard Walk
(Section~\ref{subsec:billiard}) we have to make efficient selections for the
parameter $\tau$ that controls the length of the trajectory in each step, for
the maximum number of reflections per step $\rho$, and for the walk length $W$
of the random walk. We experimentally found that setting $W=1$ the empirical
distribution converges faster to the uniform distribution. Thus, we get a higher
ESS faster than the case of $W>1$. To set $\tau$ in phase $i$, first we set
$\tau = 6\sqrt{d}r$ where $r$ is the radius of the Chebychev ball of $P_i$.
Then, we start from the center of the Chebychev ball, we perform
$100 + 4\sqrt{d}$ Billiard Walk steps and we store all the points in a set
$Q$. Then we set
$\tau = \max\{ \max\limits_{q\in Q}\{ ||q-p||_2 \}, 6\sqrt{d}r\}$. For the
maximum number of reflections we found experimentally that
$\rho = 100d$ is violated in less than $0.1\%$ of the total number of Billiard
Walk steps in our experiments.

\textbf{Rounding step}: In each phase $i$ of our method, when the minimum value of ESS among all the marginals has not reached the requested threshold, we use the generated sample to perform a rounding step by mapping the points to isotropic position. After computing the SVD decomposition of the point-set we also rescale the singular values such that the smallest one is $1$, to improve numerical stability as suggested in \cite{Cousins15}.  We found experimentally that setting the maximum number of Billiard Walk points per phase $\lambda = 20d$, where $d$ is the dimension of the polytope, suffice to improve the roundness from phase to phase. When, in any phase, the ratio between the maximum over the minimum singular value is smallest than $3$ we stop performing any new rounding step. In that case we stay on the current phase until we reach the requested value of ESS.

\begin{remark}
Given the Stoichiometric matrix $S\in\RR^{m\times n}$ of a metabolic network with flux bounds $x_{lb}\leq x\leq x_{ub}$, the total number of operations per phase that our implementation of Alg.~\ref{alg:MMCS} performs, according the parameterization given in this Section is $\OO(nd^2)$, where $d$ is the dimension of the null space of $S$ and $n$ is the number of reactions occur in the metabolic network.
\end{remark}

\subsection{Experiments}
\label{subsec:experiments}

\begin{figure}[t!]
    \centering
    \includegraphics[scale=0.293]{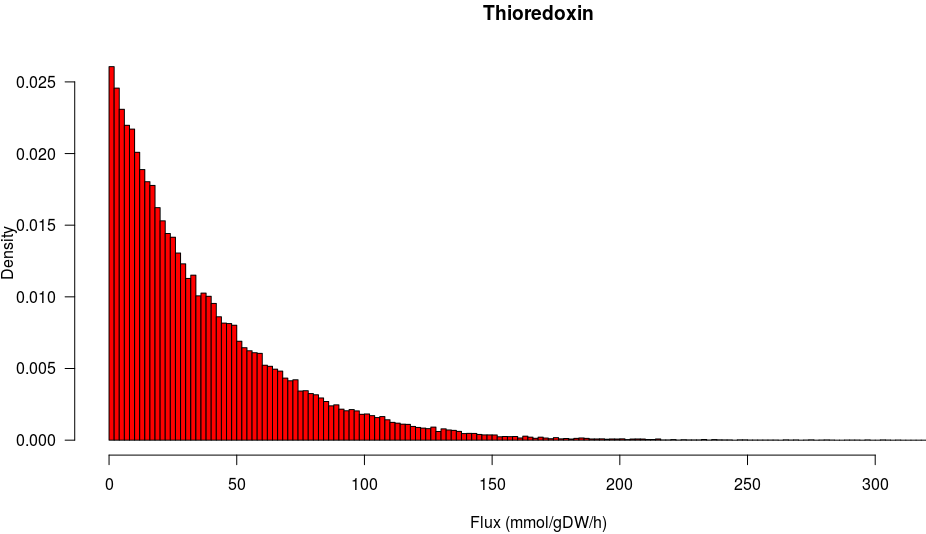}
    \includegraphics[scale=0.293]{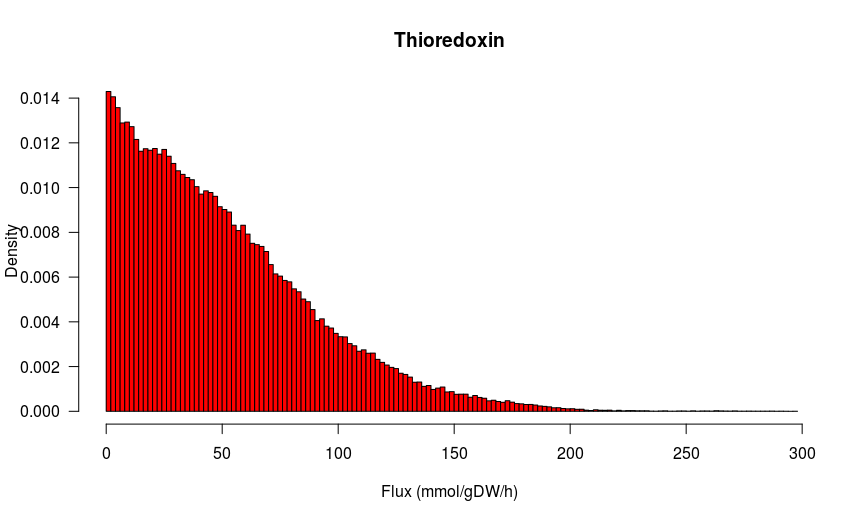}
    \caption{Our method estimation of the marginal distribution of the ``Thioredoxin reductace" flux in a constraint-based model of Homo Sapiens metabolism Recon2D~\cite{swainston16} (left) and Recon3D~\cite{brunk2018recon3d} (right).}
    \label{fig:thioredoxin}
\end{figure}

We test and evaluate our software on 17 models from the BIGG database
\cite{king2016bigg} and Recon2D, Recon3D from~\cite{Noronha18}.
In particular,
we sample from models that correspond to polytopes of dimension less than $100$;
the  simplest model in this setting is the well known bacteria
\textit{Escherichia Coli}.
We also computed with models that correspond to polytopes of  dimension a  few
thousands; this is the case for Recon2D and Recon3D.
We do not employ parallelism for
any implementation, thus we report only sequential running times.
To assess  the quality of our results we  employ a second MCMC convergence
diagnostic besides the Effective Sample Size (ESS).
This is the
potential scale reduction factor (PSRF),
introduced by Rubin and Gelman
\cite{Gelman92}.
In particular, we compute
the PSRF for each univariate marginal of the sample that MMCS outputs.
Following \cite{Gelman92},
a convergence is satisfying
according to PSRF when all the marginals have PSRF smaller than $1.1$. 

The workflow of
{\tt cobra} for sampling first performs a rounding step and then samples using Coordinate Directions Hit-and-Run (CDHR). To compare with {\tt cobra} we set the walk length of CDHR according to the empirical suggestion made in \cite{haraldsdottir2017chrr}, i.e., equal to $8d^2$, where $d$ is the dimension of the polytope we sample. For Recon2D we follow the paradigm in \cite{haraldsdottir2017chrr} which shows that the method converges for walk length equal to $1.57e$+$08$. To have a fair comparison we let {\tt cobra} to sample a minimum number of $1\, 000$ points. If in the computed sample there is a marginal with PSRF larger than $1.1$ we continue sampling until all PSRFs are smaller than $1.1$.

\begin{table}[h]
\centering
\small
\begin{tabular}{|c||c|c|c||c|c||c|c|}\hline
    &  & & & \multicolumn{2}{c}{MMCS} &   \multicolumn{2}{c}{{\tt cobra}} \\
name & (m) & (n) & (d) & Time (sec)  & (N)  & Time (sec)  & (N) \\ \hline
e\_coli\_core & 72 & 95 &  24 & 6.50e-01 & 3.40e+03 (8)   & 7.20e+01  & 4.61e+06  \\     
iLJ478 & 570 & 652 & 59  & 9.00e+00 & 5.40e+03 (5)  & 4.54e+02  & 2.79e+07 \\                   
iSB619 & 655 & 743 & 83 & 1.70e+01 &  8.20e+03 (5)  & 9.56e+02  & 5.51e+07  \\
iHN637 & 698 & 785 & 88  & 2.00e+01 & 6.80e+03 (4)  & 1.03e+03  & 6.19e+07  \\
iJN678 & 795 & 863 & 91 & 2.50e+01 & 8.10e+03 (4)  & 1.17e+03  & 6.62e+07 \\                  
iNF517  & 650 & 754 & 92  & 1.70e+01 & 6.20e+03 (4)  & 1.33e+03  & 6.77e+07  \\
iJN746  & 907 & 1054 & 116 & 5.70e+01 & 8.70e+03 (3)  & 2.22e+03  & 1.07e+08 \\              
iAB\_RBC\_283  & 342 & 469 & 130 & 5.20e+01 &  1.07e+04 (5)  & 7.85e+03  & 4.05e+08  \\
iJR904  & 761 & 1075 &  227  & 2.98e+02 & 1.62e+04 (4) & 8.81e+03  & 4.12e+08   \\
iAT\_PLT\_636   & 738 & 1008 & 289  & 3.25e+02  & 1.04e+04 (2)  & 1.73e+04  & 6.68e+08 \\
iSDY\_1059 & 1888 & 2539 & 509 & 2.813e+03 & 2.31e+04 (3) & 6.66e+04  & 2.07e+09 \\
iAF1260 & 1668 & 2382 & 516 & 6.84e+03 & 5.33e+04 (6)  & 7.04e+04  & 2.13e+09 \\
iEC1344\_C  & 1934 & 2726 & 578  & 4.86e+03 & 3.95e+04 (4) & 9.42e+04   & 2.67e+09 \\
iJO1366  & 1805 & 2583 & 582 & 6.02e+03 & 5.14e+04 (5)  & 9.99e+04  & 2.71e+09 \\
iBWG\_1329 & 1949 & 2741 & 609 & 3.06e+03 & 4.22e+04 (4) & 1.05e+05   & 2.97e+09 \\
iML1515 & 1877 & 2712 & 633 &  4.65e+03 & 5.65e+04 (5)  & 1.15e+05   & 3.21e+09 \\
Recon1 & 2766 & 3741 & 931 & 8.09e+03 & 1.94e+04 (2) & 3.20e+05  & 6.93e+09 \\
Recon2D & 5063 & 7440 & 2430 & 2.48e+04  &  5.44e+04 (2)  & $\sim 140$ days & $1.57e$+$11$   \\
Recon3D  & 8399 & 13543 & 5335 & 1.03e+05 &  1.44e+05 (2) & -- & -- \\
\hline
\end{tabular}
\caption{\label{tab:results1} 17 metabolic networks from \cite{king2016bigg} and Recon2D, Recon3D from~\cite{Noronha18}; (m) the number of Metabolites, (n) the number of Reactions, 
(d) the dimension of the polytope; (N) is the total number of sampled points $\times$ walk length;
for MMCS we stop when the sum of the minimum value of ESS among all the univariate marginals in each phase is $1000$ (we report the number of phases in  parenthesis);
for {\tt cobra} we set the walk length to $8d^2$ and $1.57e$+$08$ for Recon2D following~\cite{haraldsdottir2017chrr}, sample at least $1000$ points and stop when all marginals have PSRF $< 1.1$;
the runtime of {\tt cobra} for Recon2D is an estimation of the sequential time just for the purpose of the comparison in this paper.}
\end{table}

In Table~\ref{tab:results1} we report the results of MMCS and {\tt cobra}. We
run MMCS until we get a value of ESS equal to $1\,000$; meaning that we stop when
the sum over all phases of the minimum values of ESS among all the marginals is
larger than $1\,000$. Moreover, in Table~\ref{tab:results1} all the marginals of the sample that MMCS returns have PSRF $ < 1.1$. This is another statistical evidence on the quality of the generated sample. The histograms in Fig.~\ref{fig:thioredoxin}
illustrate an approximation for the flux distribution of the reaction
Thioredoxin as computed in Recon2D and Recon3D respectively. The same marginal flux distribution in Recon2D was estimated also in~\cite{haraldsdottir2017chrr}. Notice that the
estimated density slightly changes in Recon3D as the stoichiometric matrix has
been updated and thus the corresponding marginal is affected. In Fig.~\ref{fig:copulas} we also employ the copula representation to capture the dependence between two fluxes of reactions to confirm a mutually exclusive pair of biochemical pathways.
Notice that the run-time of MMCS is one or two orders of magnitude smaller than the run-time of {\tt cobra} and this gap becomes much larger for higher dimensional models such as Recon2D and Recon3D.

For some models --we report them in Table~\ref{tab:skip_phases}-- we introduce a
further improvement to obtain a better convergence. If there is a marginal in
the generated sample from MMCS that has a PSRF larger than $1.1$ then we do not
take into account the $k$ first phases, starting with $k=1$ until we get
both ESS equal to $1\, 000$ and all the PSRF values smaller than $1.1$ for all the
marginals. By "do not take into account" we mean that  we neither store the
generated sample --for the first $k$ phases-- nor we sum up its ESS to the
overall ESS considered for termination by MMCS. Note that for these models it is
not practical to repeat MMCS runs for different $k$ until we get the required
PSRF value. We can obtain the final results --reported in
Tables~\ref{tab:results1}-- in one pass. We simply drop a phase when the ESS
reaches the requested value but the PSRF is not smaller than $1.1$ for all the
marginals. In Table~\ref{tab:skip_phases} we separately report the MMCS runs for
different $k$ just for performance analysis reasons.

\begin{table}[h]
\centering
\begin{tabular}{|c||c|c|c|c|}\hline
\multicolumn{5}{c}{Sampling from iAF1260} \\ \hline
Phase & Avg. $\#$reflections & ESS & $\frac{\sigma_{\max}}{\sigma_{\min}}$ & Time (sec)  \\ \hline
1st & 7819 & 67 & 43459 & 2271 \\
2nd & 4909 & 68 & 922 & 1631 \\
3rd & 3863 & 77 & 582 & 1278 \\
4th & 3198 & 71 & 360 & 1080 \\
5th & 1300 & 592 & 29 & 454 \\
6th & 1187 & 4821 & 3.5 & 417 \\
7th & 1181 & 4567 & 2.8 & 415 \\
\hline
\end{tabular}
\caption{We sample $20d = 10320$ points per phase with Billiard Walk and walk length equal to $1$, where $d = 516$ is the dimension of the corresponding polytope. For each phase we report the average number of reflections per BIlliard Walk step, the the minimum value of Effective Sample Size among all the univariate marginals, the ratio between the maximum over the minimum singular value derived from the SVD decomposition of the generated sample and the run-time.\label{tab:sampling_phases}}
\end{table}

Interestingly, the total number of Billiard Walk steps --and consequently the
run-time-- does not increase as $k$ increases in Table~\ref{tab:skip_phases}.
This means that the performance of our method improves for these models, when we
do not take into account the $k$ first phases of MMCS. This happens because the
performance of Billiard Walk improves as the polytope becomes more rounded from
phase to phase. In particular, in Table~\ref{tab:sampling_phases} we analyze the
performance of
Billiard Walk for the model iAF1260. We sample $20d$ points per
phase with walk length equal to $1$ and we report the average number of
reflections, the ESS, the run-time, and the ratio $\sigma_{\max} / \sigma_{\min}$
per phase. The latter is the ratio between the maximum over the minimum singular
value of the point-set. The larger this ratio is the more skinny the polytope of
the corresponding phase is. As the method progresses from the first to the last
phase, the average number of reflections and the run-time decrease and the ESS
increases. This means that as the polytope becomes more rounded from phase to
phase, the Billiard Walk step becomes faster and the generated sample has better
quality. This explains why the total run-time does not increase when we do not
take into account the first  $k$ phases: the initial phases are slow and
they  contribute poorly to the  quality of the final sample; the last phases are
fast and contribute with more accurate samples.



\bibliography{metabolic_networks}


\appendix

\section{The Billiard walk algorithm}

\begin{algorithm}[H]
  \caption{Billiard Walk$(P, p, \rho, \tau, W)$}
	\label{alg:billiard}
	\SetKwInOut{Input}{Input}
	\SetKwInOut{Output}{Output}
    \SetKwInOut{Require}{Require}
    \SetKwRepeat{Do}{do}{while}

    \Input{polytope $P$; point $p$; upper bound on the number of reflections $\rho$; length of trajectory parameter $\tau$; walk length $W$.}

    \Require { point $p\in P$}

	\Output{A point in $P$}

    \BlankLine
    \For {$j=1,\dots ,W$} {
       $L \leftarrow -\tau\ln\eta$, $\eta\sim \mathcal{U}(0,1)$ \tcp{length of the trajectory}

       $i\leftarrow 0$ \tcp{current number of reflections}

       $p_0\leftarrow p$ \tcp{initial point of the step}

       pick a uniform vector $\varv_0$ from the boundary of the unit ball

    	\Do{$i\leq \rho$} {
    		 $\ell \leftarrow \{p_i + t\varv_i, 0\leq t\leq L\}$ \tcp{segment}
    		\If{$\partial P\cap\ell=\varnothing$} {
    			$p_{i+1} \leftarrow p_i+L\varv_i$ \;
          \textbf{break} \;
    		}
        
    		$p_{i+1} \leftarrow \partial P\cap\ell$ ; \tcp{point update}
    		the inner vector, $s$, of the tangent plane at $p$, s.t.\ $||s|| = 1$\;
    		$L \leftarrow L - |P\cap\ell|$\;
    		$\varv_{i+1} \leftarrow \varv_i - 2(\varv_i^Ts) s$ \tcp{direction update}
    		$i \leftarrow i+1$\;
    	}
    	\leIf{$i=\rho$} {
    		$p \leftarrow p_0$
    	}
    	{
    		$p \leftarrow p_i$
    	}
    }
    \KwRet $p$\;
\end{algorithm}

\section{Additional experiments}\label{sec:apndx}

\begin{table}[h]
\centering
\begin{tabular}{c||c|c|c|c||c|c|c|c||c|c|c|c||c|c|c|c||c|c|c|c||c|c|c|c|}\hline
 \multicolumn{5}{c}{Sampling from iAF1260}\\
\hline
 Do not take into account the sample of the k first phases & Time (sec) & PSRF < 1.1 &  (M) & (N)  \\ \hline
0 first phases & 6955 & $41\%$ & 6 & 56100\\
1 first phases & 6943 & $56\%$ & 6 & 54100  \\
2 first phases & 6890 & $76\%$ & 6 & 55200 \\
3 first phases & 6867 & $95\%$ & 6 & 53200  \\
4 first phases & 6840 & $100\%$ & 6 & 53300 \\
\hline
 \multicolumn{5}{c}{Sampling from iBWG\_1329}\\
\hline

0 first phases & 3067 & $50\%$ & 4 & 42100\\
1 first phases & 3189 & $97\%$ & 5 & 48800 \\
2 first phases & 4652 & $100\%$ & 5 & 56500 \\
\hline
 \multicolumn{5}{c}{Sampling from iEC1344\_C}\\
\hline

0 first phases & 4845 & $77\%$ & 4 & 41100 \\
1 first phases & 4721 & $96\%$ & 4 & 42500 \\
2 first phases & 4682 & $100\%$ & 4 & 39500 \\
\hline
 \multicolumn{5}{c}{Sampling from iJO1366}\\
\hline

0 first phases & 3708 & $66\%$ & 5 & 51500 \\
1 first phases & 6022 & $100\%$ & 5 & 51400 \\
\end{tabular}
\caption{We run our method and we do not take into account the sample of the $k$ first phases, thus we do not also count the value of the Effective Sample Size (ESS) in those phases, before we start storing the generated sample and sum up the ESS of each phase. In all cases MMCS stops when the sum of ESS reaches $1000$. For each case we report the total run-time, the percentage of the marginals that have PSRF smaller than $1.1$ and the total number of phases (M) generates included the $k$ first phases and the total number of Billiard Walk steps (N) included those performed in the $k$ first phases.\label{tab:skip_phases}}
\end{table}

\end{document}